%% file: robot_tiles.tex
\documentclass[english]{article}
\usepackage[utf8]{inputenc} 
\usepackage[english]{babel} 
\usepackage[noend]{algpseudocode}
\usepackage[hyphens]{url}
\usepackage[pdfborder={0 0 0}]{hyperref}
\usepackage{cite}
\usepackage{graphicx}
\usepackage{float}
\usepackage{caption}
\usepackage{subfigure}
\usepackage{amsmath,amssymb, amsthm}
\usepackage{thm-restate}
\usepackage{geometry}
\usepackage{todonotes}
\usepackage{xspace}
\usepackage{microtype}
\usepackage{pdfpages}
\usepackage{color}
\usepackage{authblk}

\newboolean{full}
\setboolean{full}{true}
\newcommand{\fullversion}[1]{\ifthenelse{\boolean{full}}{#1}{}}
\newcommand{\shortversion}[1]{\ifthenelse{\boolean{full}}{}{#1}}

\floatstyle{boxed}
\newfloat{algorithm}{htpb}{ext}
\floatname{algorithm}{Algorithm}

\newtheorem{definition}{Definition}
\newtheorem{theorem}[definition]{Theorem}

\newtheorem{corollary}[definition]{Corollary}

\graphicspath{{figures/}}
\newcommand{\OCal}{\ensuremath{\mathcal{O}}}
\newcommand{\ignore}[1]{}

\title{CADbots: Algorithmic Aspects of Manipulating Programmable Matter with Finite Automata}

	\author[1]{Sándor P. Fekete}
	\author[2]{Robert Gmyr}
	\author[1]{Sabrina Hugo}
	\author[1]{Phillip Keldenich}
	\author[1]{Christian Scheffer}
	\author[1]{Arne Schmidt}
\date{}
		
		\affil[1]{Department of Computer Science, TU Braunschweig, Germany.\newline
			\tt{\{fekete, hugo, keldenich, scheffer, aschmidt\}@ibr.cs.tu-bs.de}}
		\affil[2]{Department of Computer Science, University of Houston, USA
			\tt{rgmyr@uh.edu}}
		
\begin{document}
	\maketitle

\begin{abstract}
We contribute results for a set of fundamental problems in the context of programmable matter 
by presenting algorithmic methods for evaluating and manipulating
a collective of particles by a finite automaton
that can neither store significant amounts of data,
nor perform complex computations, 
and is limited to a handful of possible physical operations.
We provide a toolbox for carrying out fundamental tasks
on a given arrangement of tiles, using the arrangement itself as a
storage device, similar to a higher-dimensional Turing machine with geometric
properties. Specific results include time- and space-efficient procedures for 
bounding, counting, copying, reflecting, rotating or scaling a complex given shape.
	\end{abstract}

	\input{01-introduction.tex}
	\input{02-preliminaries.tex}

	\input{03-basic-tools.tex}

	\input{04-counting.tex}

\input{05-turing.tex}

\input{06-functions.tex}

\input{07-conclusion.tex}

\bibliography{bib}
\bibliographystyle{abbrv}
\newpage

\end{document}

%% file: 01-introduction.tex

\section{Introduction}
When dealing with classic challenges of robotics, such as exploration,
evaluation and manipulation of objects, traditional robot models are based
on relatively powerful capabilities, such as the ability (1) to collect
and store significant amounts of data, (2) perform intricate computations,
and (3) execute complex physical operations. With the ongoing progress in
miniaturization of devices, new possibilities emerge for exploration,
evaluation and manipulation. However,
dealing with micro- and even nano-scale dimensions (as present in the context of
programmable matter) introduces a vast spectrum of
new difficulties and constraints. These include significant limitations to
all three of the mentioned capabilities; challenges get even more pronounced
in the context of complex nanoscale systems, where there is a significant threshold
between ``internal'' objects and sensors and ``external'' control entities
that can evaluate gathered data, extract important information and provide guidance.

In this paper, we present algorithmic methods for evaluating and manipulating
a collective of particles by agents of the most basic possible type: finite automata
that can neither store significant amounts of data, nor perform complex computations, and are limited to a handful of possible physical operations.
The objective is to provide a toolbox for carrying out fundamental tasks
on a given arrangement of particles, such as bounding, counting, copying,
reflecting, rotating or scaling a large given shape. The key idea is to use
the arrangement itself as a storage device, similar to a higher-dimensional
Turing machine with geometric properties.

\subsection{Our Results}
We consider an arrangement $P$ of $N$ pixel-shaped particles, on which a single finite-state robot
can perform a limited set of operations; see Section~\ref{sec:prelim} for a precise model description.
Our goal is to develop strategies for evaluating and modifying the arrangement by defining sequences of
transformations and providing metrics for such sequences. In particular, we present the following;
a full technical overview is given in Table~\ref{tab:results}.

\begin{itemize}
\item We give a time- and space-efficient method for determining the bounding box of $P$.
\item We show that we can simulate a Turing machine with our model.
\item We provide a counter for evaluating the number $N$ of tiles forming $P$, as well as the number
of corner pixels of $P$.
\item We develop time- and space-efficient algorithms for higher-level operations, such as copying,
reflecting, rotating or scaling $P$.
\end{itemize}

\begin{table}[!ht]
	\resizebox{1\columnwidth}{!}{\renewcommand*{\arraystretch}{1.1}
	\begin{tabular}{|p{2.65cm}||c|c|c|}
		\hline
		\textbf{Problem} & \textbf{Tile Complexity} & \textbf{Time Complexity} & \textbf{Space Complexity}\\
		\hline\hline
		Bounding Box & $\OCal(|\partial P|)$ & $\OCal(wh\max(w,h))$ & $\OCal(w+h)$ \\
		\hline\hline
	    Counting: &&&\\
		\hfill $N$ Tiles &$\OCal(\log N)^{*}$&$\OCal\left(\max(w,h)\log N + N\min(w,h)\right)$&$\OCal(\max(w,h))$\\
		\hfill $k$ Corners &$\OCal(\log k)^{*}$&$\OCal\left(\max(w,h)\log k + k\min(w,h)+wh\right)$&$\OCal(\max(w,h))$\\
		\hline\hline
		Function:&&&\\
		\hfill Copy & $\OCal(N)^{*}$ & $\OCal(wh^2)$ & $\OCal(wh)$ \\
		\hfill Reflect & $\OCal(\max(w,h))^{*}$ & $\OCal((w+h)wh)$ & $\OCal(w+h)$ \\
		\hfill Rotate &  $\OCal(w+h)^{*}$ & $\OCal((w+h)wh)$ & $\OCal(|w-h|\max(w,h))$ \\
		\hfill Scaling by $c$ &  $\OCal(cN)$ & $\OCal((w^2+h^2)c^2N)$ & $\OCal(cwh)$
		\\\hline
	\end{tabular}
 	}\vspace{0.2cm}
 	\caption{Results of this paper. $N$ is the number of tiles in the given shape $P$, $w$ and $h$ its  width and height. ($^{*}$) is the number
of auxiliary tiles after constructing the bounding box.}
 	\label{tab:results}
\vspace*{-9mm}
\end{table}

\subsection{Related Work}

There is a wide range of related work; due to space limitations, we can only present a
small selection. 

A very successful model considers self-assembling DNA tiles
(e.g., \cite{TileAssemblySurvey,StagedSelfAssembly})
that form complex shapes based on the interaction of different glues along their edges; however, no active
agents are involved, and composition is the result of chemical and physical diffusion.

The setting of a finite-automaton robot operating on a set of tiles in a grid was introduced in~\cite{GmyrDNA18}, where
the objective is to arrange a given set of tiles into an equilateral triangle.
An extension studies the problem of recognizing certain shapes~\cite{GmyrMFCS18}.
We use a simplified variant of the model underlying this line of research that exhibits three main differences:
First, for ease of presentation we consider a square grid instead of a triangular grid.
Second, our model is less restrictive in that the robot can create and destroy tiles at will instead of only being able to transport tiles from one position to another.
Finally, we allow the robot to move freely throughout the grid instead of restricting it to move along the tile structure.

Other previous related work includes shape formation in a number of different models:
 in the context of agents walking DNA-based shapes~\cite{ReifS09, Anupama17,Wickham12};
in the Amoebot model~\cite{AmoebotShapeFormation}; in modular
robotics~\cite{ModularRobotics}; in a variant of population
protocols~\cite{PopulationProtocolShapeFormation}; in the nubot model~\cite{NubotShapeFormation}.
Models based on automate and movable objects have also been studied in the
context of one-dimensional arrays, e.g., pebble automata~\cite{PebbleAutomata}.
Work focusing on a setting of robots on graphs includes network exploration
\cite{NetworkExplorationProblem}, maze exploration~\cite{MazeExploration},
rendezvous search~\cite{RobotsOnGraphsRendezvous}, intruder caption and graph
searching \cite{CopsAndRobbers,GraphSearching}, and black hole search~\cite{BlackHoleSearch}.
For a connection to other approaches to agents
moving tiles, e.g., see~\cite{PushingBlocks,Assembly}.

Although the complexity of our model is very restricted, actually realizing
such a system, for example using complex DNA nanomachines, is currently still a
challenging task. However,
on the practical side, recent years have seen significant progress towards realizing
systems with the capabilities of our model.
For example, it has been shown that nanomachines have the ability
to act like the head of a finite automaton on an input tape~\cite{ReifS09}, to
walk on a one- or two-dimensional
surface~\cite{DNALandscape,Omabegho09,Wickham12}, and to transport
cargo~\cite{MolecularTransport,Anupama17,WangEW12}.

%% file: 02-preliminaries.tex
\section{Preliminaries}\label{sec:prelim}
We consider a single \emph{robot} that acts as a \emph{deterministic finite automaton}. 
The robot moves on the (infinite) \emph{grid} $G=(\mathbb{Z}^2, E)$ with edges between all pairs of nodes that are within unit distance. 
The nodes of $G$ are called \emph{pixels}.
Each pixel is in one of two \emph{states}: It is either \emph{empty} or \emph{occupied} by a particle called \emph{tile}.
A \emph{diagonal pair} is a pair of two pixels $(x_1,y_1),(x_2,y_2)$ with $|x_1-x_2|=|y_1-y_2| = 1$ (see Fig.~\ref{fig:prelim}, left and middle).
A \emph{polyomino} is a connected set of tiles $P \subset \mathbb{Z}^2$ such that for all diagonal pairs $p_1,p_2 \in P$ there is another tile $p \in P$ that is adjacent to $p_1$ and adjacent to $p_2$ (see Fig.~\ref{fig:prelim} middle and right). 
We say that $P$ is \emph{simple} if it has no holes, i.e., if there is no empty pixel $p \in G$ that lies inside a cycle formed by occupied tiles.
Otherwise, $P$ is \emph{non-simple}.

The \emph{$a$-row} of $P$ is the set of all pixels $(x,a) \in P$. 
We say that $P$ is \emph{$y$-monotone} if the $a$-row of $P$ is connected in $G$ for each $a \in \mathbb{Z}$ (see Fig.~\ref{fig:prelim} right). 
Analogously, the \emph{$a$-column} of $P$ is the set of all pixels $(a,y) \in P$ and $P$ is called \emph{$x$-monotone} if the $a$-column of $P$ is connected in $G$ for each $a \in \mathbb{Z}$.
The \emph{boundary} $\partial P$ of $P$ is the set of all tiles of $P$ that are adjacent to an empty pixel or build a diagonal pair with an empty pixel (see Fig.~\ref{fig:prelim} right).

\begin{figure}[t]
	\centering\hfill
	\includegraphics[width=0.17\columnwidth]{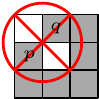}\hfill
	\includegraphics[width=0.17\columnwidth]{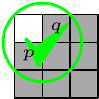}\hfill
	\includegraphics[width=0.4\columnwidth]{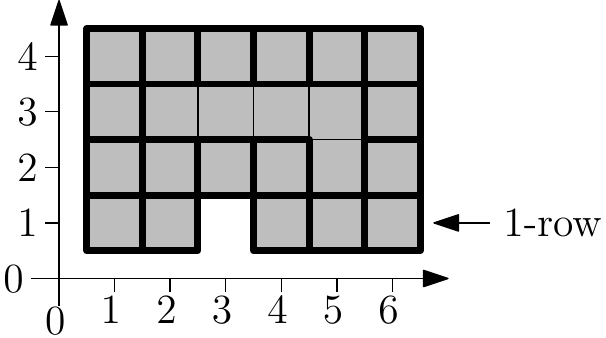}\hfill\phantom{}
	\caption{
		Left: An illegal diagonal pair $(p,q)$.
		Middle: An allowed diagonal pair $(p,q)$.
		Right: A polyomino $P$ with its boundary $\partial P$ (tiles with bold outline).
		$P$ is $x$-monotone but not $y$-monotone, because the $1$-row is not connected.%
	}
	\label{fig:prelim}
\end{figure}

A \emph{configuration} consists of the states of all pixels and the robot's location and state.
The robot can transform a configuration into another configuration using a sequence of look-compute-move steps as follows.
In each step, the robot acts according to a \emph{transition function} $\delta$.
This transition function maps a pair $(p,b)$ containing the state of the robot and the state of the current pixel to a triple $(q,c,d)$, where $q$ is the new state of the robot, $c$ is the new state of the current pixel, and $d \in \{\textrm{up},\textrm{down},\textrm{left},\textrm{right}\}$ is the direction the robot moves in.
In other words, in each step, the robot checks its state $p$ and the state of the current pixel $b$, computes $(q,c,d) = \delta(p,b)$, changes into state $q$, changes the state of the current pixel to $c$ if $c \neq b$, and moves one step into direction $d$.

Our goal is to develop robots transforming an unknown initial configuration $S$ into a target configuration $T(S)$.
We assess the efficiency of a robot by several metrics: 
(1) \emph{Time complexity}, i.e., the total number of steps performed by the robot until termination,
(2) \emph{space complexity}, i.e., the total number of visited pixels outside the bounding box of the input polyomino $P$, and
(3) \emph{tile complexity}, i.e., the maximum number of tiles on the grid minus the number of tiles in the input polyomino $P$.

%% file: 03-basic-tools.tex
\section{Basic Tools}
A robot can check the states of all pixels within a constant distance by performing a tour of constant length. Thus, from now on we assume that a robot can check within a single step all eight pixels that are adjacent to the current position $r$ or build a diagonal pair with $r$.

\subsection{Bounding Box Construction}
\label{sec:bbox}
We describe how to construct a bounding box of a given polyomino $P$ in the form of a zig-zag as shown in Fig.~\ref{fig:bb_example}. 
We can split the bounding box into an \textit{outer lane} and an \textit{inner lane} (see Fig.~\ref{fig:bb_example}).
This allows distinguishing tiles of $P$ and tiles of the bounding box.

Our construction proceeds in three phases.
(i) We search for an appropriate starting position.
(ii) We wrap a zig-zag path around $P$.
(iii) We finalize this path to obtain the bounding box.

\ignore{
The main idea is to wrap a zig-zag line around the polyomino and push the line
outwards if there occurs a conflict with the polyomino, i.e., tiles of the line
share a corner or a side with a tile of $P$ (in case the second tile belongs to
the line, we have completely enclosed $P$).  A further problem while
constructing this line is to determine at which point we have to make a turn.
We observe that after a turn, the next turn cannot be made before finding a tile that has distance two to the actual position (or else we can create a cycle not containing $P$). 
If such a tile has been found, we do a turn when there is no such tile anymore.
}
\begin{figure}
	\centering
  	\includegraphics[width=0.5\columnwidth]{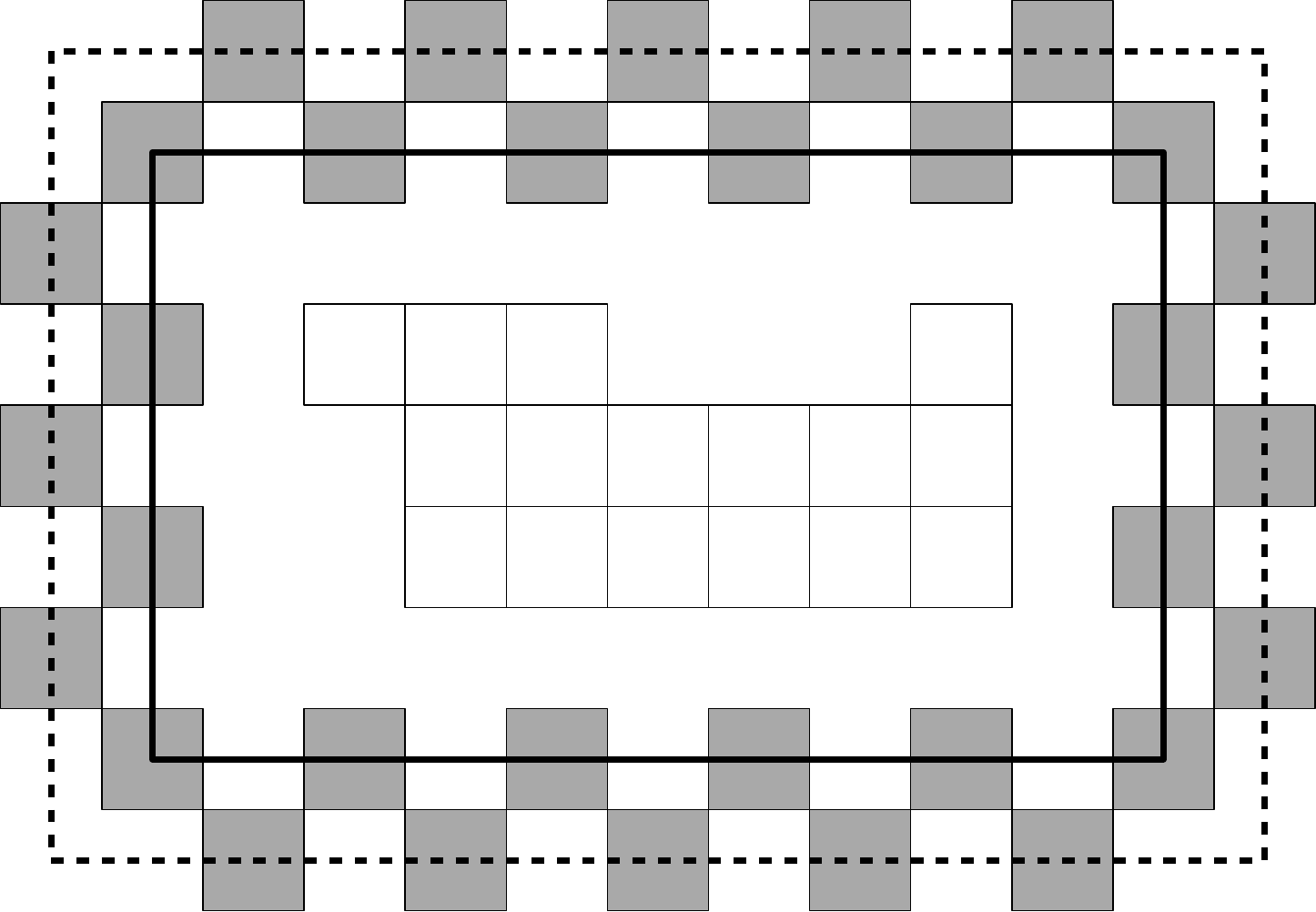}
	\caption{A Polyomino $P$ (white) surrounded by a bounding box of tiles (gray) on the inner lane (solid line) and tiles on the outer lane (dashed line).}
	\label{fig:bb_example}
	\vspace{-0.3cm}
\end{figure}
 
 \ignore{
We construct the bounding box as follows:
First, we find a starting position. This can be done by searching for a local minimum in the y-direction. In case of ties, we take the one with minimal x-coordinate, i.e., the leftmost tile.
}
For \textbf{phase (i)}, we simply search for a local minimum in the $y$-direction. 
We break ties by taking the tile with minimal $x$-coordinate, i.e., the leftmost tile.
From the local minimum, we go two steps further to the left.
If we land on a tile, this tile belongs to $P$ and we restart phase (i) from this tile.
Otherwise we have found a possible \emph{start position}.

In \textbf{phase (ii)}, we start constructing a path that will wrap around $P$.
We start placing tiles in a zig-zag manner in the upwards direction.
While constructing the path, three cases may occur.
\begin{enumerate}
	\item At some point we lose contact with $P$, i.e., there is no tile at distance two from the inner lane.
	In this case, we do a right turn and continue our construction in the new direction.

	\item Placing a new tile produces a conflict, i.e., the tile shares a corner or a side with a tile of $P$.
	In this case, we shift the currently constructed side of the bounding box outwards until no more conflict occurs.
	This shifting process may lead to further conflicts, i.e., we may not be able to further shift the current side outwards.
	If this happens, we deconstruct the path until we can shift it outwards again (see Fig.~\ref{fig:bb_conflicts}).
	In this process, we may be forced to deconstruct the entire bounding box we have built so far, i.e., we remove all tiles of the bounding box including the tile in the starting position.
	In this case we know that there must be a tile of $P$ to the left of the start position; we move to the left until we reach such a tile and restart phase (i).
	
	\item Placing a new tile closes the path, i.e., we reach a tile of the bounding box we have created so far.
	We proceed with phase (iii).
\end{enumerate}

\ignore{
During the construction we may get conflicts, i.e., the bounding box can hit $P$.
In this case, we can shift the current bounding box line in the direction $d$, which is the direction after a left turn.
While shifting the line, we may encounter further conflicts.
If this is the case, we have to remove tiles of the bounding box until no more conflict occurs and only then proceed with the shifting.
(Note that in some cases we completely have to remove the construction built so far. We know that to left of our starting position has to be a tile of $P$. Therefore, we can go left until we reach this tile and restart the whole process.)
}

\textbf{Phase (iii):} Let $t$ be the tile that we reached at the end of phase (ii).
We can distinguish two cases: (i) At $t$, the bounding box splits into three different directions
(shown as the tile with black-and-white diagonal stripes in Fig.~\ref{fig:bb_cases}),
and (ii) the bounding box splits into two different directions. 
In the latter case we are done, because we can only reach the bottom or left side of the bounding box.
In the first case we have reached a right or top side of the bounding box.
From $t$ we can move to the tile where we started the bounding box construction,
and remove the bounding box parts until we reach $t$.
Now the bounding box splits in two directions at $t$.
Because we have reached a left or top side of the bounding box, we may not have built a convex shape around $P$ (see Fig.~\ref{fig:bb_cases} left).
This can be dealt with by straightening the bounding box in an analogous fashion, e.g., by pushing the line to the right of $t$ down. 

\begin{figure}
	\vspace{-0.1cm}
	\centering	\hfill
	\includegraphics[width=0.22\columnwidth]{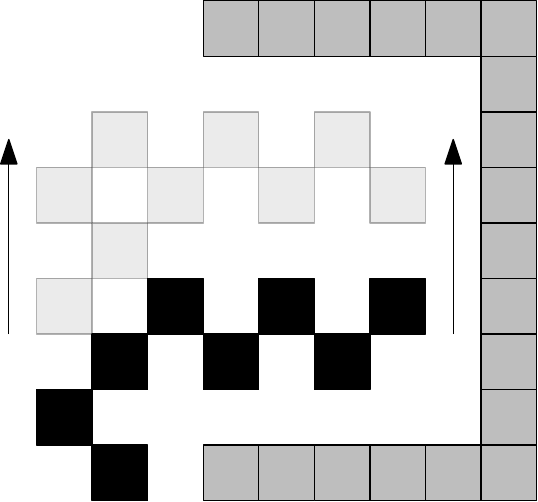}\hfill
	\includegraphics[width=0.22\columnwidth]{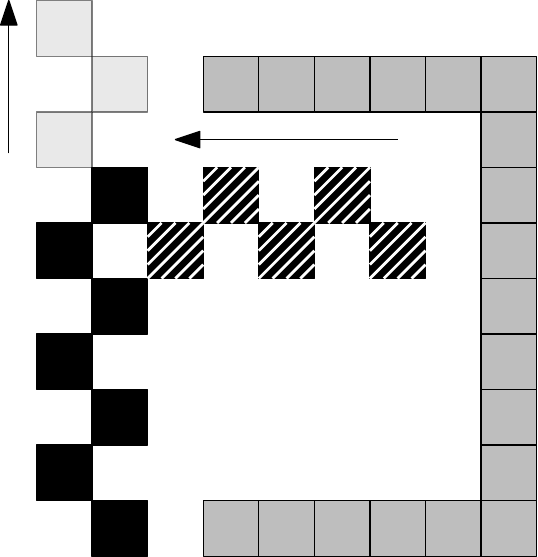}\hfill\phantom{}
	\caption{Left: Further construction is not possible. Therefore, we shift the line upwards (see light gray tiles).
		Right: A further shift produces a conflict with the polyomino. Thus, we remove tiles (diagonal stripes) and proceed with the shift (light gray) when there is no more conflict.}
	\label{fig:bb_conflicts}
	\vspace{-0.05cm}
\end{figure}
\begin{figure}
	\centering
	\hfill
	\includegraphics[width=0.43\columnwidth]{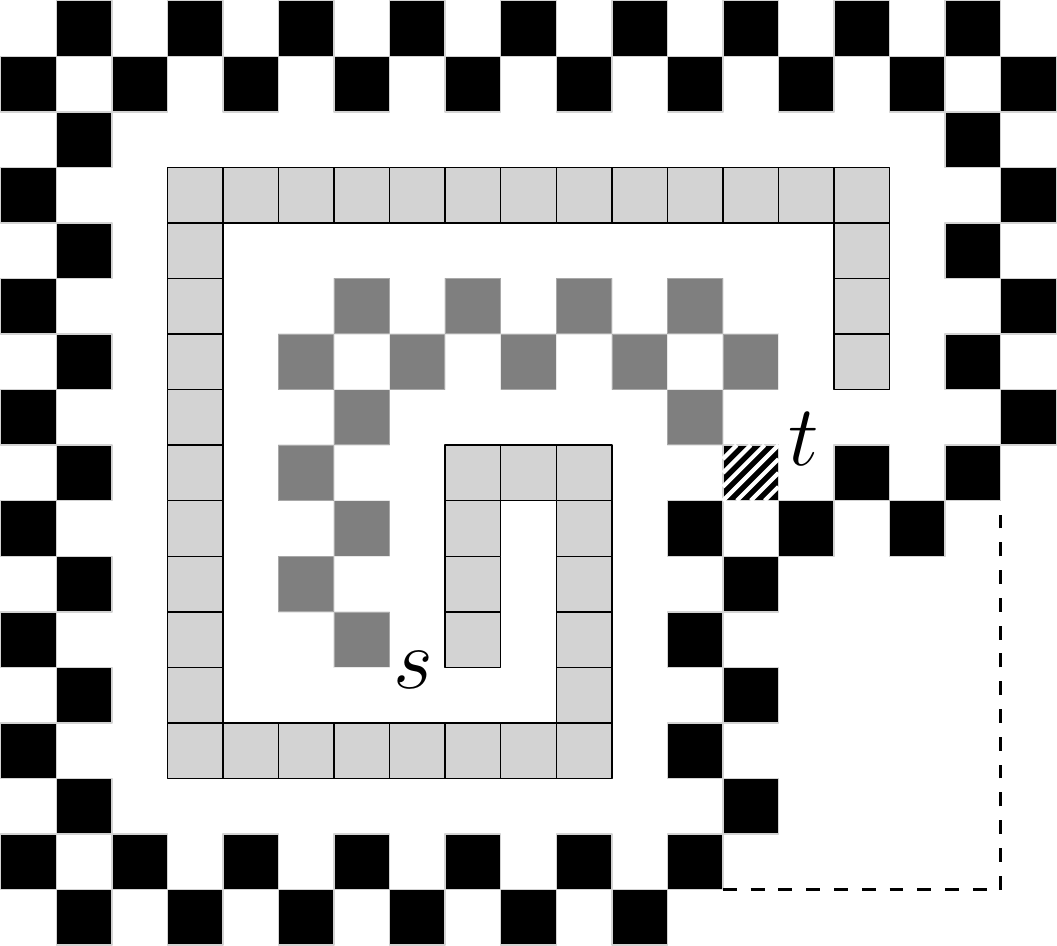}\hfill
	\includegraphics[width=0.43\columnwidth]{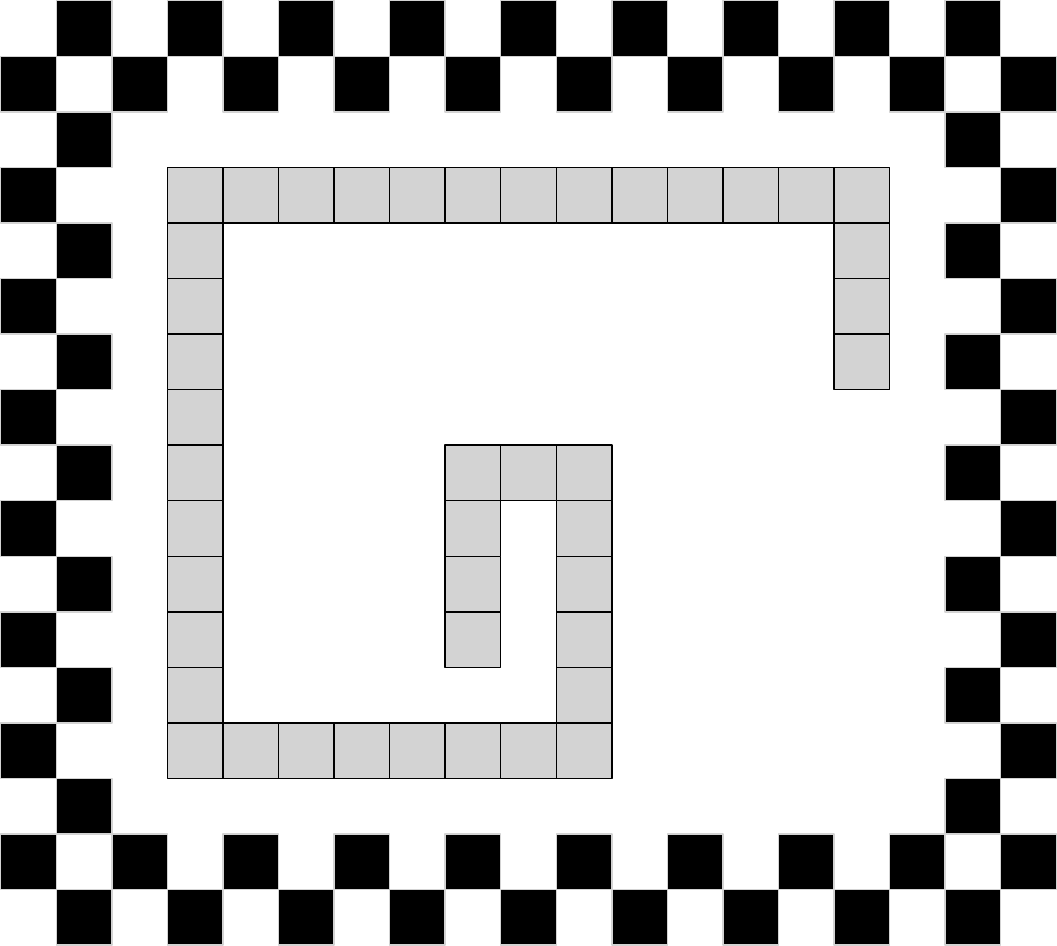}\hfill\phantom{}
	\caption{Left: During construction, starting in $s$, we reach a tile
$t$ (diagonal stripes) at which the bounding box is split into three
directions. The part between $s$ and $t$ (dark gray tiles) can be removed,
followed by straightening the bounding box along the dashed line. Right: The final
bounding box.} \label{fig:bb_cases}
\vspace{-0.4cm}
\end{figure}

\begin{theorem}\label{thm:bounding_box}
	The described strategy  builds a bounding box that encloses the given polyomino $P$ of width $w$ and height $h$. 
	Moreover, the strategy terminates after $\OCal(\max(w,h)\cdot wh)$ steps, using at most $\OCal(|\partial P|)$ auxiliary tiles and only $\OCal(w+h)$ of additional space. The running time in the best case is $\OCal(wh)$.
\end{theorem}

\begin{proof}
	\shortversion{We can show that each pixel is visited a constant number of times, when the correct start has been found implying a running time of $\OCal(wh)$.
	Also, each tile of the bounding box can be charged to a tile of $\partial P$ such that each such tile has constant cost implying a tile complexity of $\OCal(\partial P)$.
	The details of the proof can be found in the appendix, see Section~\ref{app:bounding_box}.}
	\fullversion{
	\textbf{Correctness:}
	We show that (a) the bounding box will enclose $P$; and (b) whenever we make a turn or shift a side of the bounding box, we find a tile with distance two from the bounding box, i.e., we will not build an infinite line.
	
	First note that we only make a turn after encountering a tile with distance two to the inner lane.
	This means that we will ``gift wrap'' the polyomino until we reach the start.
	When there is a conflict (i.e., we would hit a tile of $P$ with the current line), we shift the current line.
	Thus, we again find a tile with distance two to the inner lane.
	This also happens when we remove the current line and extend the previous one. 
	After a short extension we make a turn and find a tile with distance two to the inner lane, meaning that we make another turn at some point.
	Therefore, we do not construct an infinite line, and the construction remains finite.
	
	\textbf{Time:}
	To establish a runtime of $\OCal(wh)$, we show that each pixel lying in the final bounding box is visited only a constant number of times.
	Consider a tile $t$ that is visited for the first time.
	We know that $t$ gets revisited again if the line through gets shifted or removed.
	When $t$ is no longer on the bounding box, we can visit $t$ again while searching for the start tile.
	Thus, $t$ is visited at most four times by the robot, implying a running time of $\OCal(wh)$ unit steps.
	However, it may happen that we have to remove the bounding box completely and have to restart the local minimum search.
	In this case, there may be tiles that can be visited up to $\max(w,h)$ times (see Fig.~\ref{fig:bb_worst_case}).
	Therefore, the running time is $\OCal(\max(w,h)\cdot wh)$ in the worst-case.
	\begin{figure}
		\centering	
		\includegraphics[width=0.8\columnwidth]{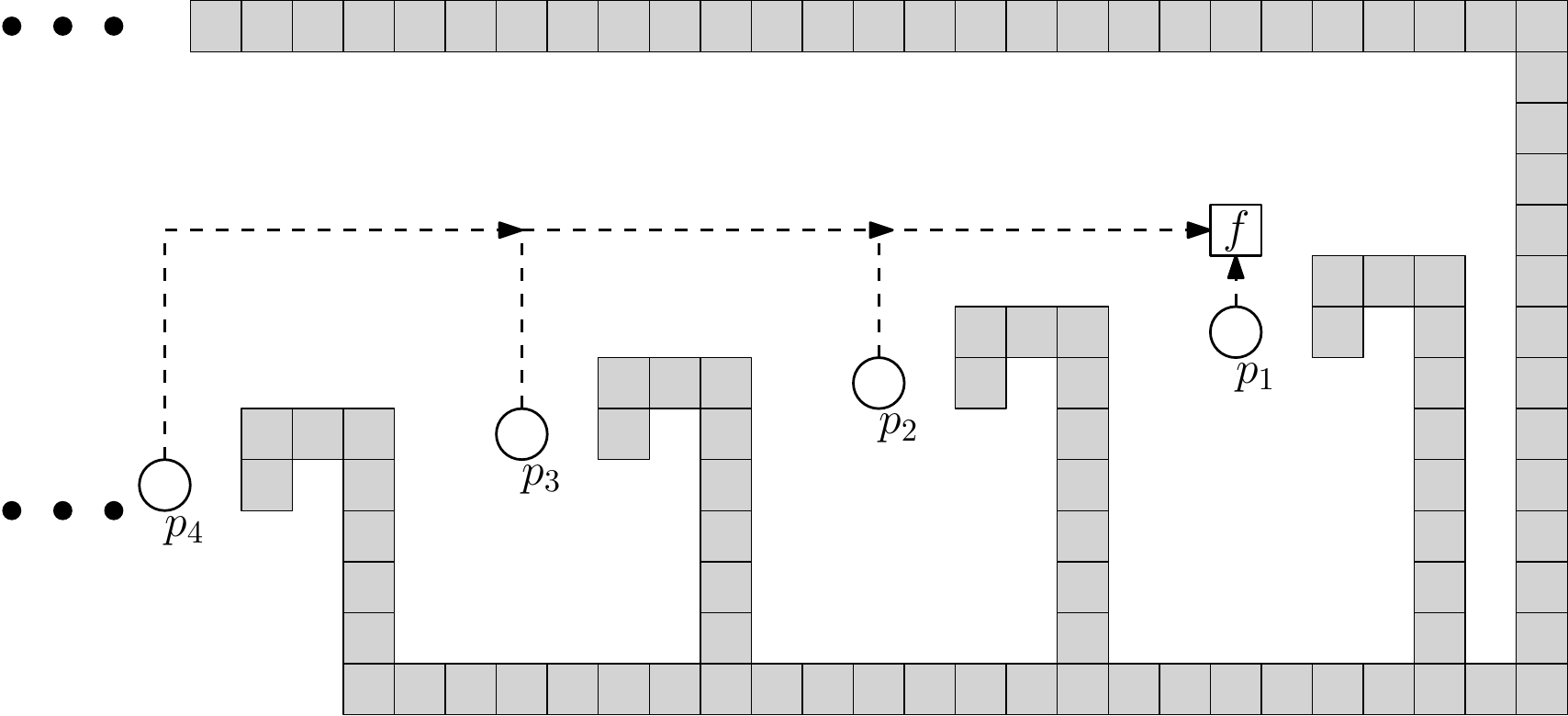}
		\caption{A worst-case example for building a bounding box.
			The positions $p_1$ to $p_4$ denote the first, second, third and fourth starting position of the robot during bounding box construction.
			With each restart, the pixel $f$ is visited at least once.
			Therefore, $f$ can be visited $\Omega(w)$ times.}
		\label{fig:bb_worst_case}
	\end{figure}
	
	\textbf{Auxiliary Tiles:}
	We now show that we need at most $\OCal(|\partial P|)$ many auxiliary tiles at any time.
	Consider a tile $t$ of $\partial P$ from which we can shoot a ray to a tile of the bounding box (or the intermediate construction), such that no other tile of $P$ is hit.
	For each tile $t$ of $\partial P$, there are at most four tiles $t_1,\dots, t_4$ of the bounding box.
	We can charge the cost of $t_1,\dots, t_4$ to $t$, which is constant. 
	Thus, each tile in $\partial P$ has been charged by $\OCal(1)$.
	However, there are still tiles on the bounding box that have not been charged to a tile of $\partial P$, i.e., tiles that lie in a curve of the bounding box.
	
	Consider a locally convex tile $t$, i.e., a tile at which we can place a $2\times 2$ square solely containing $t$. 
	Each turn of the bounding box can be charged to a locally convex tile.
	Note that there are at most four turns that can be charged to a locally convex tile.
	For each turn, there are at most four tiles that are charged to a locally convex tile.
	Thus, each tile of $\partial P$ has constant cost, i.e., we need at most $\OCal(\partial P)$ auxiliary tiles. 
	
	It is straightforward to see that we need $\OCal(w+h)$ additional space.}
\end{proof}

\subsection{Binary Counter}
For counting problems, a binary counter is indispensable, because we are not able to store non-constant numbers.
The binary counter for storing an $n$-bit number consists of a base-line of $n$ tiles.
Above each tile of the base-line there is either a tile denoting a 1, or an empty pixel denoting 0.
Given an $n$-bit counter we can \textit{increase} and \textit{decrease} the counter by 1 (or by any constant $c$), and \textit{extend} the counter by one bit.
The latter operation will only be used in an increase operation.
In order to perform an increase or decrease operation, we move to the least significant bit and start flipping the bit (i.e., remove or place tiles).

\begin{figure}
	\centering
	\includegraphics[width=0.55\columnwidth]{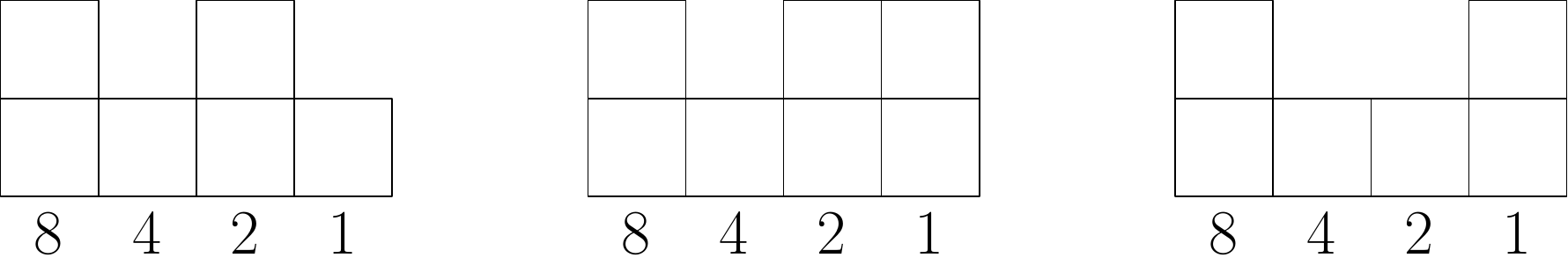}
	\caption{Left: A $4$-bit counter with decimal value 10 (1010 in binary). Middle: The binary counter increased by one. Right: The binary counter decreased by one.}
	\label{fig:binary_counter}
\end{figure}

%% file: 04-counting.tex
\section{Counting Problems}
The constant memory of our robot poses several challenges when we want to count certain elements of our polyomino, such as tiles or corners.
Because these counts can be arbitrarily large, we cannot store them in the state space of our robot.
Instead, we have to make use of a binary counter.
This requires us to move back and forth between the polyomino and the counter.
Therefore, we must be able to find and identify the counter coming from the polyomino and to find the way back to the position where we stopped counting.

This motivates the following strategy for counting problems.
We start by constructing a slightly extended bounding box.
We perform the actual counting by shifting the polyomino two units downwards or to the left, one tile at a time.
After moving each tile, we return to the counter and increase it if necessary.
We describe further details in the next two sections, where we present algorithms for counting the number of tiles or corners in a polyomino.

\subsection{Counting Tiles}
The total number of tiles in our polyomino can be counted using the strategy outlined above, increasing the counter by one after each moved tile.
\begin{theorem}\label{th:count_tiles}
	Let $P$ be a polyomino of width $w$ and height $h$ with $N$ tiles for which the bounding box has already been created.
	Counting the number of tiles in $P$ can be done in $\OCal(\max(w,h)\log N + N\min(w,h))$ steps using $\OCal(\max(w,h))$ of additional space and $\OCal(\log N)$ auxiliary tiles.
\end{theorem}

\begin{proof}
	\begin{figure}[t]
		\centering
		\includegraphics[width=0.8\linewidth]{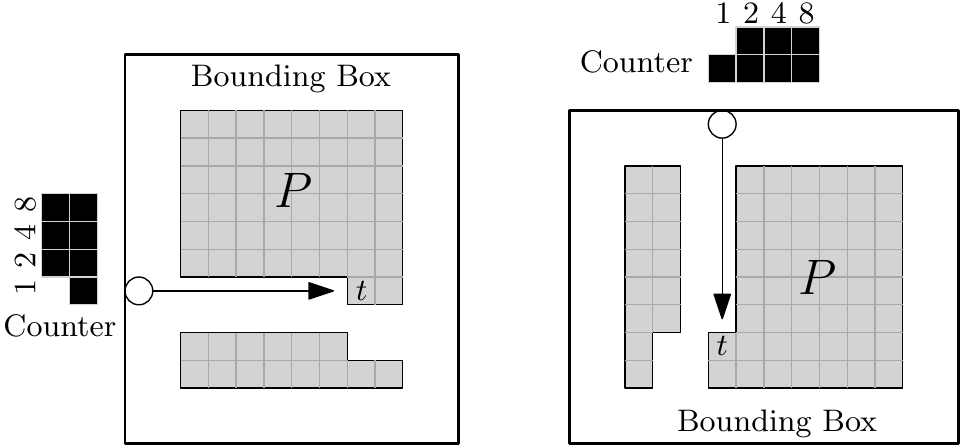}
		\caption{
			An $8 \times 8$ square $P$ in a bounding box, being counted row-by-row (left) or column-by-column (right).
			The robot has already counted 14 tiles of $P$ and stored this information in the binary counter.
			The arrow shows how the robot ($\bigcirc$) moves to count the next tile $t$.
		}
		\label{fig:counting_example}
	\end{figure}
	
In a first step, we determine whether the polyomino's width is greater than its height or vice versa.
We can do this by moving to the lower left corner of the bounding box and then alternatingly moving up and right one step until we meet the bounding box again.
We can recognize the bounding box by a local lookup based on its zig-zag shape that contains tiles only connected by a corner, which cannot occur in a polyomino.
The height is at least as high as the width if we end up on the right side of the bounding box.
In the following, we describe our counting procedure for the case that the polyomino is higher than wide; the other case is analogous.

We start by extending the bounding box by shifting its bottom line down by two units.
Afterwards we create a vertical binary counter to the left of the bounding box.
We begin counting tiles in the bottom row of the polyomino.
We keep our counter in a column to the left of the bounding box such that the least significant bit at the bottom of the counter is in the row, in which we are currently counting.
We move to the right into the current row until we find the first tile.
If this tile is part of the bounding box, the current row is done. In this case, we move back to the counter, shift it upwards and continue with the next row until all rows are done.
Otherwise, we simply move the current tile down two units, return to the counter and increment it.
For an example of this procedure, refer to Fig.~\ref{fig:counting_example}.

For each tile in the polyomino, we use $\OCal(\min(w,h))$ steps to move to the tile, shift it and return to the counter; incrementing the counter itself only takes $\OCal(1)$ amortized time per tile.
For each empty pixel we have cost $\OCal(1)$.
In addition, we have to shift the counter $\max(w,h)$ times.
Thus, we need $\OCal(\max(w,h)\log N + \min(w,h)N+wh)=\OCal(\max(w,h)\log N + \min(w,h)N)$ unit steps.
We only use $\OCal(\log N)$ auxiliary tiles in the counter, in addition to the bounding box.

In order to achieve $\OCal(1)$ additional space, we modify the procedure as follows.
Whenever we move our counter, we check whether it extends into the space above the bounding box.
If it does, we reflect the counter vertically, such that the least significant bit is at the top in the row, in which we are currently counting.
This requires $\OCal(\log^2 N)$ extra steps and avoids using $\OCal(\log N)$ additional space above the polyomino.
\end{proof}

\subsection{Counting Corners}
In this section, we present an algorithm to count reflex or convex corners of a given polyomino.

\begin{theorem}
	Let $P$ be a polyomino of width $w$ and height $h$ with $n$ convex (reflex) corners for which the bounding box has already been created.
	Counting the number of $k$ convex (reflex) corners in $P$ can be done in $\OCal(\max(w,h)\log k + k\min(w,h)+$ 
	$wh)$ steps, using $\OCal(\max(w,h))$ of additional space and $\OCal(\log k)$ auxiliary tiles.
\end{theorem}
\begin{figure}[t]
	\centering
	\hfil
	\includegraphics[scale=0.4]{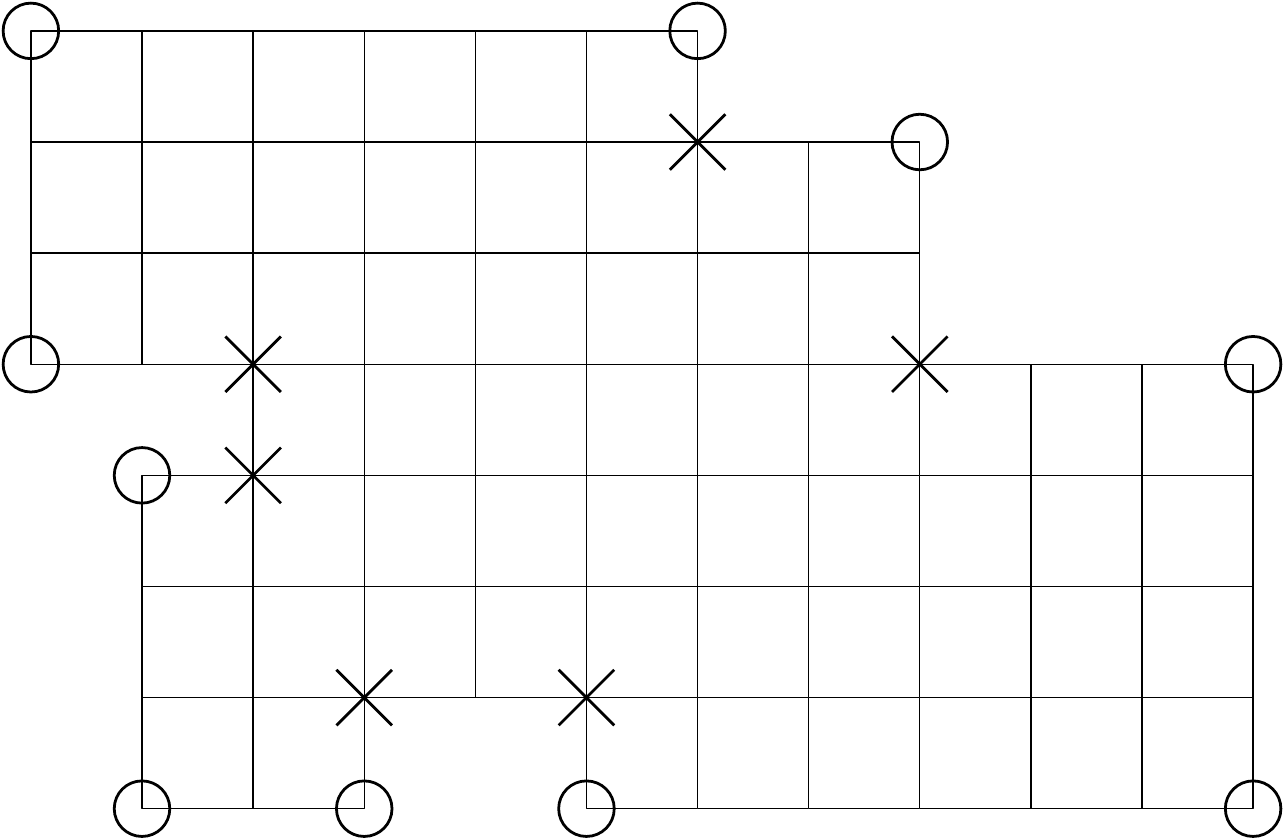}
	\hfil
	\includegraphics[scale=0.5]{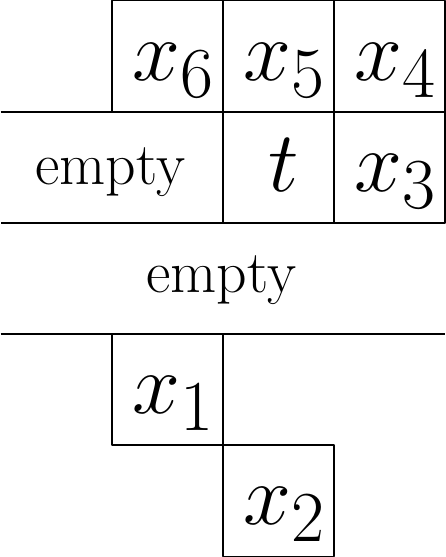}\hfil\hfil
	
	\caption{
		\emph{Left:} A polyomino and its convex corners ($\circ$) and reflex corners ($\times$).
		\emph{Right:} After reaching a new tile $t$, we have to know which of the pixel $x_1$ to $x_6$ are occupied to decide how many convex (reflex) corners $t$ has.}
	\label{fig:counting_tiles}
\end{figure}
\begin{proof}
	We use the same strategy as for counting tiles in Theorem~\ref{th:count_tiles} and retain the same asymptotic bounds on running time, auxiliary tiles and additional space.
	It remains to describe how we recognize convex and reflex corners.

	Whenever we reach a tile $t$ that we have not yet considered so far, we examine its neighbors $x_1,\ldots,x_6$, as shown in Fig.~\ref{fig:counting_tiles}.
	The tile $t$ has one convex corner for each pair $(x_1,x_2)$, $(x_2,x_3)$, $(x_3,x_5)$, $(x_5,x_1)$ that does not contain a tile, and no convex corner is contained in more than one tile of the polyomino.
	As there are at most four convex corners per tile, we can simply store this number in the state space of our robot, return to the counter, and increment it accordingly.

	A reflex corner is shared by exactly three tiles of the polyomino, so we have to ensure that each corner is counted exactly once.
	We achieve this by only counting a reflex corner if it is on top of our current tile $t$ and was not already counted with the previous tile $x_1$.
	In other words, we count the upper right corner of $t$ as a reflex corner if exactly two of $x_3,x_4,x_5$ are occupied; we count the upper left corner of $t$ as reflex corner if $x_6$ and $x_5$ are present and $x_1$ is not.
	In this way, all reflex corners are counted exactly once.
\end{proof}

%% file: 05-turing.tex
\section{Transformations with Turing machines}
In this section we develop a robot that transforms a polyomino~$P_1$ into a required target polyomino $P_2$. In particular, we encode $P_1$ and $P_2$ by strings $S(P_1)$ and $S(P_2)$ whose characters are from $\{ 0,1,\sqcup\}$ (see Fig.~\ref{fig:TM} left and Definition~\ref{def:bin_repr}). If there is a Turing machine transforming $S(P_1)$ into $S(P_2)$, we can give a strategy that transforms $P_1$ into $P_2$ (see Theorem~\ref{thm:TM_strategy}).

We start with the definition of the encodings $S(P_1)$ and $S(P_2)$.

\begin{figure}[t]
	\begin{center}\hfill
	\includegraphics[width=0.25\columnwidth]{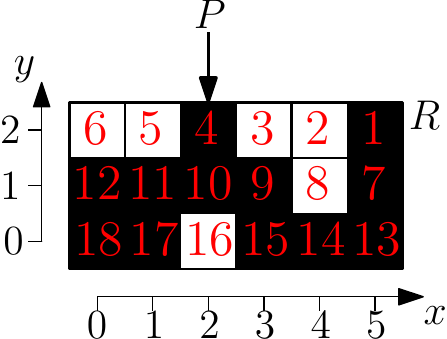}\hfill
	\includegraphics[width=0.68\columnwidth]{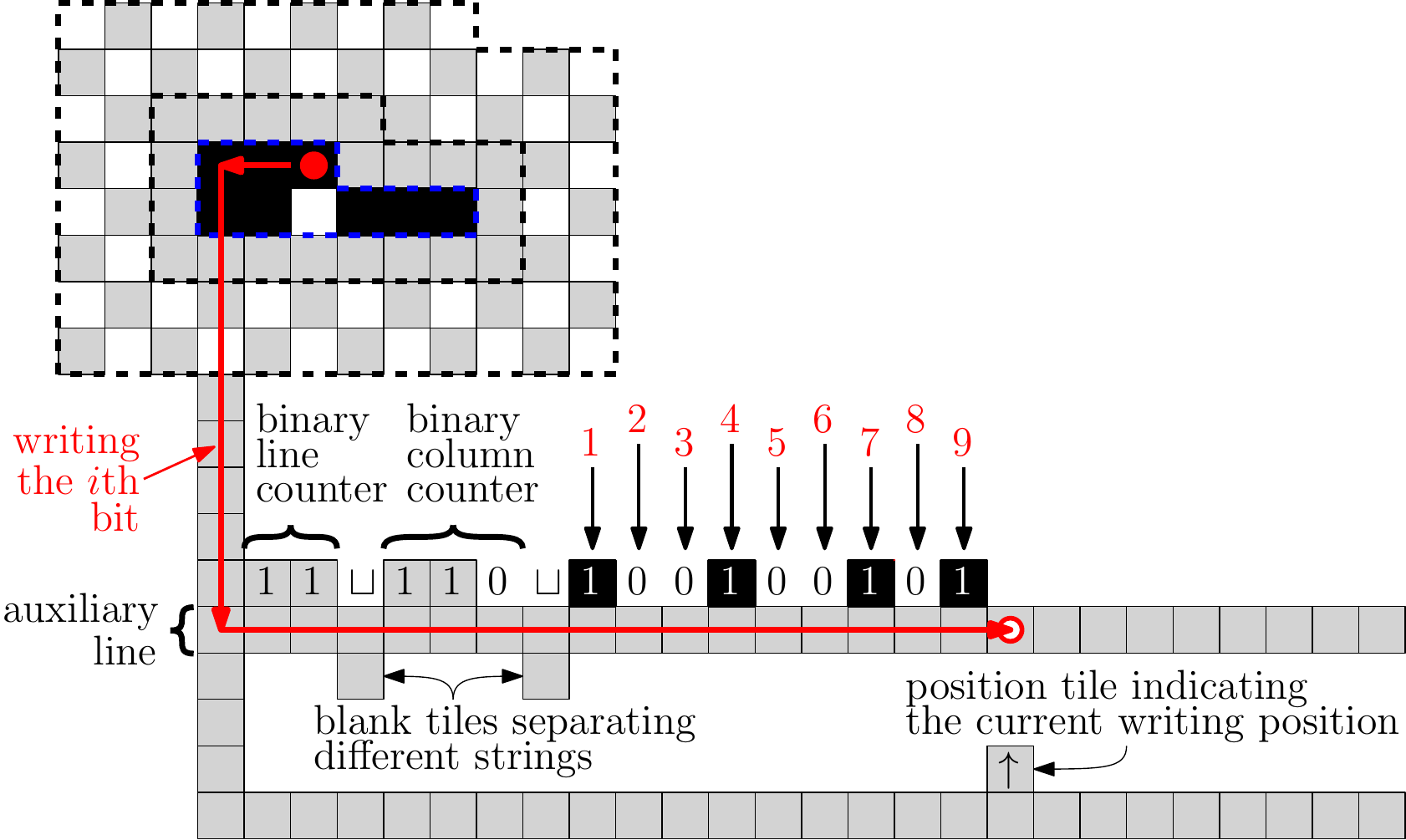}\hfill\phantom{}
	\end{center}
	\caption{Left: \ignore{The string encoding $S(P)= S_1 \sqcup S_2 \sqcup S_3 =11 \sqcup 110 \sqcup 100100101111111011$ of a black polyomino $P$. $S_1$~and $S_2$ are the height and the width of the smallest enclosing rectangle~$R$ of $P$. $S_3$ indicates whether the $i$th tile of $R$ belongs to $P$ or not.}
		 An example showing how to encode a polyomino of height $h=3$ ($S_1=11$ in binary) and width $w=6$ ($S_2=110$ in binary) by $S(P)=S_1 \sqcup S_2 \sqcup S_3=11 \sqcup 110 \sqcup 100100101111111011$, where $S_3$ is the string of 18 bits that represent tiles (black, 1 in binary) and empty pixel (white, 0 in binary), proceeding from high bits to low bits.
		 Right: Phase (2) writing $S(P)$ (currently the $10$th bit) on a horizontal auxiliary line.}
	\label{fig:TM}
	
\end{figure}

\begin{definition}\label{def:bin_repr}
	Let $R:=R(P)$ be the polyomino forming the smallest rectangle containing a given polyomino $P$ (see Fig.~\ref{fig:TM} right). We represent $P$ by the concatenation $S(P):=S_1 \sqcup S_2 \sqcup S_3$ of three bit strings $S_1$, $S_2$, and $S_3$ separated by blanks~$\sqcup$. In particular, $S_1$ and $S_2$ are the height and width of $R$. Furthermore, we label each tile of $R$ by its rank in the lexicographic decreasing order w.r.t. $y$- and $x$-coordinates with higher priority to $y$-coordinates (see Fig.~\ref{fig:TM} right).  Finally, $S_3$ is an $|R|$-bit string $b_1,\dots,b_{|R|}$ with $b_i = 1$ if and only if the $i$th tile in $R$ also belongs to $P$.
	\end{definition}
	

\begin{theorem}\label{thm:TM_strategy}
	Let $P_1$ and $P_2$ be two polyominos with $|P_1| = |P_2| = N$. 
	There is a strategy transforming $P_1$ into $P_2$ if there is a Turing machine transforming $S(P_1)$ into $S(P_2)$. 
	The robot needs $\mathcal{O}(\partial P_1 + \partial P_2 + S_{TM})$ auxiliary tiles, $\mathcal{O}(N^4 + T_{TM})$ steps, and $\Theta (N^2 + S_{TM})$ of additional space, where $T_{TM}$ and $S_{TM}$ are the number of  steps and additional space needed by the Turing machine.
\end{theorem}
\begin{proof}
	\begin{figure}[htbp]
		\centering
		\subfigure[State after counting numbers of lines and columns.]{\includegraphics[width=0.74\textwidth]{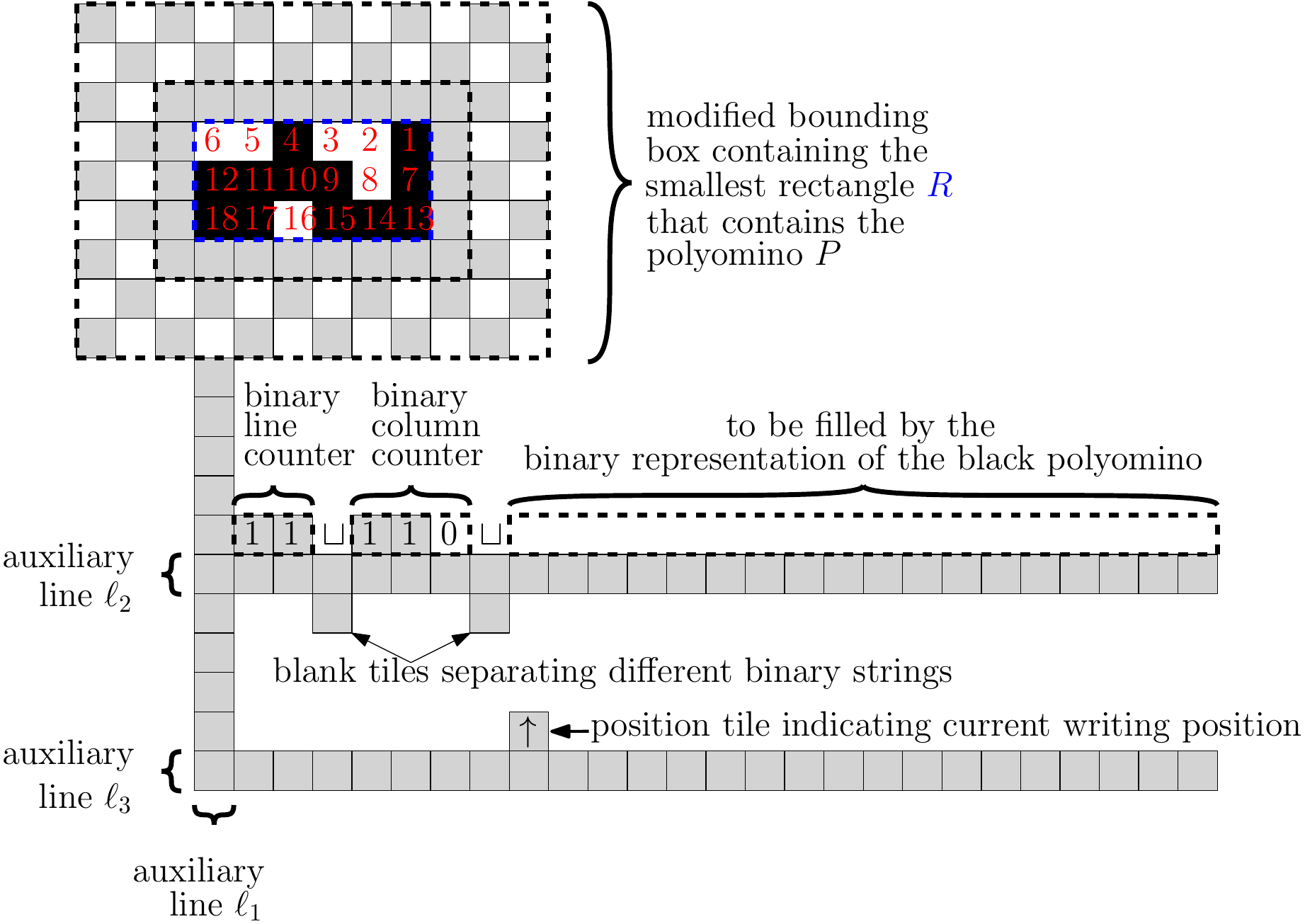}}\label{fig:TMa}
		\subfigure[Intermediate state of Phase (2) while writing the string representation $S(P_1)$ of the input polyomino $P_1$ by placing $R$ line by line onto $\ell_1$.] {\includegraphics[width=0.7\textwidth]{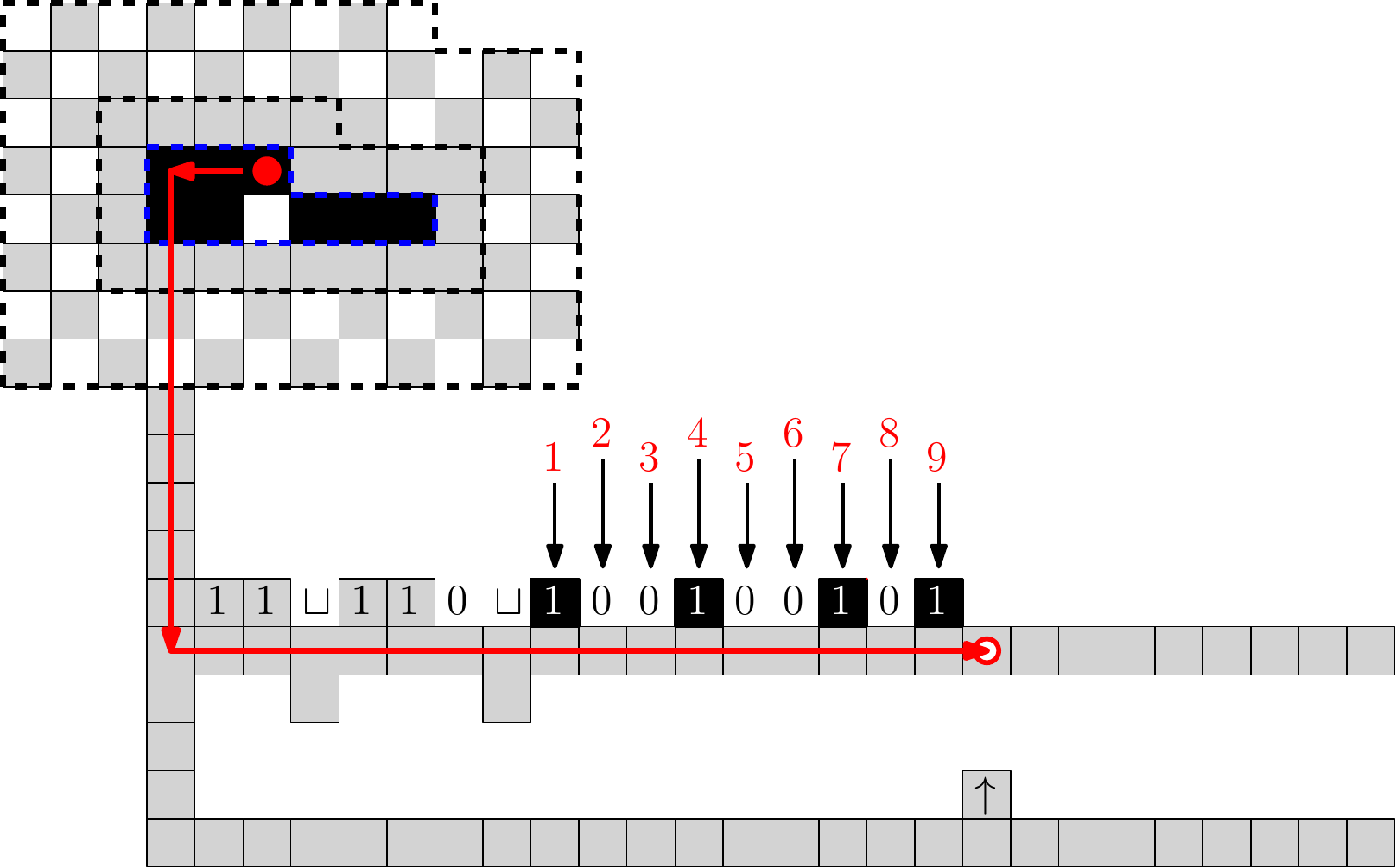}}\label{fig:TMb}
		\subfigure[Final state of Phase (2) after writing the binary representation of the input polyomino.]{\includegraphics[width=0.7\textwidth]{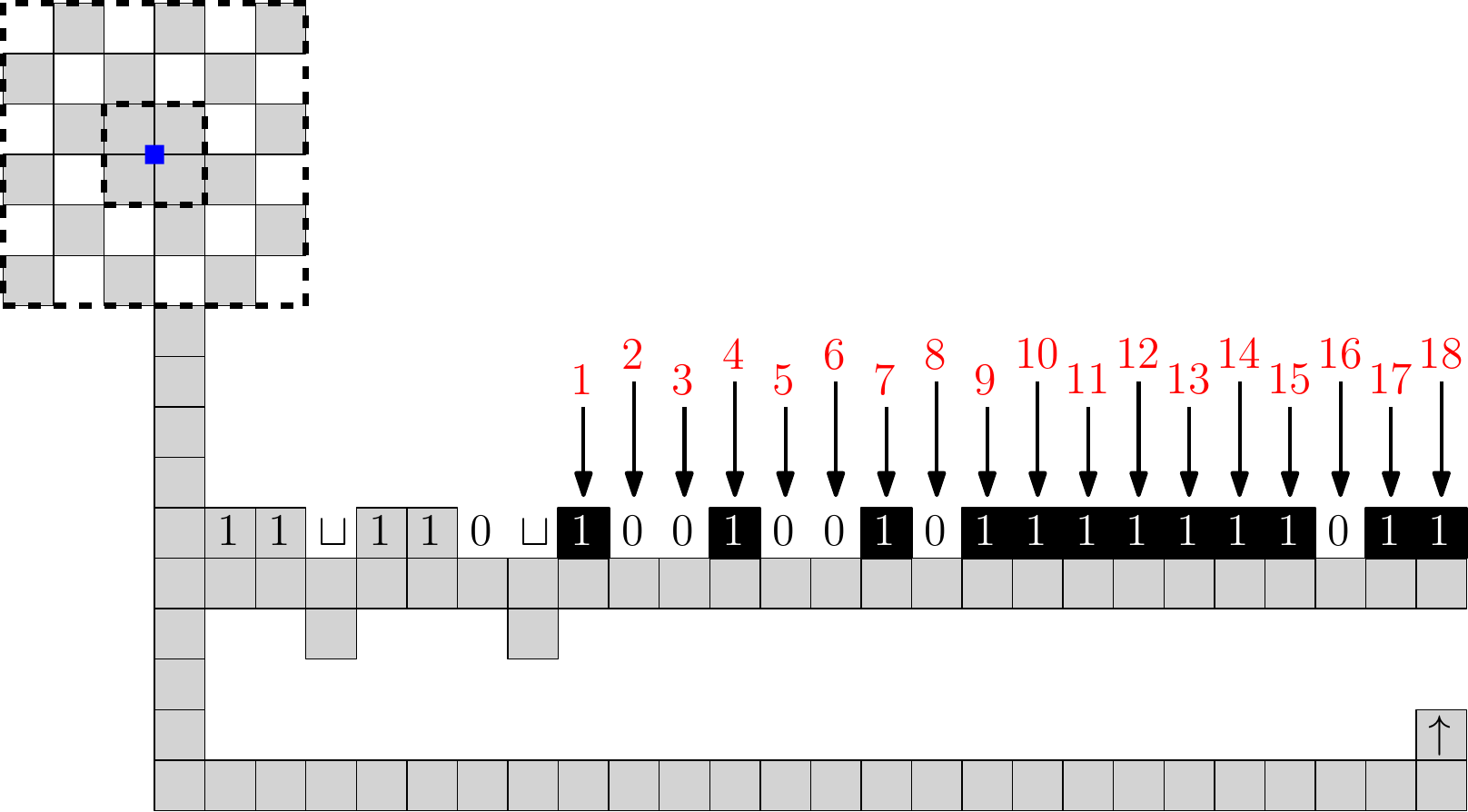}}\label{fig:TMc}
		\caption{Phase (2) of transforming a (black) polyomino by simulating a Turing Machine.} \label{fig:TM2}
	\end{figure}
	Our strategy works in five phases: Phase (1) constructs a slightly modified bounding box of $P_1$ (see Fig.~\ref{fig:TM}(a)). 
	Phase (2) constructs a shape representing $S(P_1)$. 
	In particular, the robot writes $S(P_1)$ onto a auxiliary line $\ell_2$ in a bit-by-bit fashion (see Fig~\ref{fig:TM}(b)). 
	In order to remember the position, to which the previous bit was written, we use a \emph{position tile} placed on another auxiliary line $\ell_3$.
	Phase (3) simulates the Turing machine that transforms $S(P_1)$ into $S(P_2)$. Phase (4) is the reversed version of Phase (2), and Phase (5) is the reversed version of Phase (1). As Phase~(4) and Phase (5) are the reversed version of Phase (1) and Phase (2), we still have to discuss how to realize Phases~(1),~(2), and~(3).
	
	\textbf{Phase (1):} We apply a slightly modified version of the approach of Theorem~\ref{thm:bounding_box}. 
	In particular, in addition to the bounding box, we fill all pixels that lie not inside $R$ and adjacent to the boundary of $R$ by auxiliary tiles. 
	This results in a third additional layer of the bounding box, see Fig.~\ref{fig:TM}(a). 
	Note that the robot recognizes that the tiles of the third layer do not belong to the input polyomino $P$, because these tiles of the third layer lie adjacent to the original bounding box. 
	
	\textbf{Phase (2):} Initially, we construct a vertical auxiliary line $\ell_1$ of constant length seeding two horizontal auxiliary lines $\ell_2$ and $\ell_3$ that are simultaneously constructed by the robot during Phase~(2). 
	The robot constructs the representation of $S(P_1)$ on $\ell_2$ and stores on $\ell_3$ a ``position tile'' indicating the next tile on $\ell_2$, where the robot has to write the next bit of the representation of~$S(P_1)$.
	
	Next we apply the approach of Theorem~\ref{th:count_tiles} twice in order to construct on $\ell_2$ one after another the binary representations of the line and column numbers of $R$, see Fig.~\ref{fig:TM}.
	
	Finally, the robot places one after another and in decreasing order w.r.t. $y$-coordinates all lines of $R$ onto $\ell_2$. In particular, the tiles of the current line $\ell$ of $R$ are processed from right to left as follows. For each each tile $t \in \ell$, the robot alternates between $\ell$ and the auxiliary lines $\ell_2$ and $\ell_3$ in order to check whether $t$ belongs to~$P_1$ or not.
	
	In order to reach the first tile of the topmost line, the robot moves in a vertical direction on the lane induced by the vertical line $\ell_1$ to the topmost line of $R$ and then to the rightmost tile of the topmost line, see the yellow arrows in Fig.~\ref{fig:TM}(b). 
	In order to ensure that in the next iteration of Phase (2), the next tile of $R$ is reached, the robot deletes $t$ from $R$ and shrinks  the modified bounding box correspondingly, see Fig.~\ref{fig:TM}(b). 
	This involves only a constant-sized neighbourhood of $t$ and thus can be realized by robot with constant memory. 
	Next, the robot moves to the first free position of auxiliary line $\ell_2$ by searching for the position tile on auxiliary line~$\ell_3$. As the vertical distances of $\ell_2$ and $\ell_3$ to the lowest line of the bounding box is a constant, the robot can change vertically between $\ell_2$ and $\ell_3$ by simply remembering the robot's vertical distance to the lowest line of the bounding box. 
	The robot concludes the current iteration of Phase (2) by moving the position tile one step to the right.
	
	The approach of Phase (2) is repeated until all tiles of $P_1$ are removed from $R$.
	
	\textbf{Phase (3):} We simply apply the algorithm of the Turing machine by moving the robot correspondingly to the movement of the head of the Turing machine.
	
	Finally, we analyze the entire approach of simulating a Turing machine as described above. 
	From the discussion of Phase (2) we conclude that the construction of the auxiliary lines $\ell_1$, $\ell_2$, and $\ell_3$ are not necessary, because the vertical length of the entire construction below the bounding is a constant. 
	Furthermore, the current writing position on $\ell_2$ is indicated by a single position tile on $\ell_2$. 
	Thus, sizes of the bounding boxes needed in Phase (2) and Phase (4) dominate the number of used auxiliary tiles, which is in $\mathcal{O}(\partial P_1 + \partial P_2)$.
	Furthermore, each iteration of Phase (2) needs $\mathcal{O}(N^4 + T_{TM})$ steps, where $T_{TM}$ is the number of steps needed by the Turing machine.
	The same argument applies to the number of steps used in Phase (4).
	Finally, the additional space is in~$\Theta (N^2)$. 
	This concludes the proof of Theorem~\ref{thm:TM_strategy}.
\end{proof}

%% file: 06-functions.tex
\section{CAD Functions}
As already seen, we can transform a given polyomino by any computable function by simulating a Turing machine.
However, this requires a considerable amount of space.
In this section, we present realizations of common functions in a more space-efficient manner.

\subsection{Copying a Polyomino}
Copying a polyomino $P$ has many application.
For example, we can apply algorithms that may destroy $P$ after copying $P$ into free space.
In the following, we describe the \emph{copy} function that copies each column below the bounding box, as seen in~Fig.~\ref{fig:copy_bridge} (a row-wise copy can be described analogously).

Our strategy to copy a polyomino is as follows:
After constructing the bounding box, we copy each column of the bounding box area into free space.
This, however, comes with a problem.
As the intersection of the polyomino with a column may consists of more than one component, which may be several units apart, we have to be sure that the right amount of empty pixels are copied.
This problem can be solved by using \textit{bridges}, i.e., auxiliary tiles denoting empty pixels. 
To distinguish between tiles of $P$ and bridges, we use two lanes (i) a left lane containing tiles of $P$ and (ii) a right line containing bridge tiles (see Fig.~\ref{fig:copy_bridge}).

To copy an empty pixel, we place a tile two unit steps to the left. This marks a part of a bridge. 
Then we move down and search for the first position, at which there is no copied tile of $P$ to the left, nor a bridge tile.
When this position is found, we place a tile, denoting an empty position.

To copy a tile $t$ of $P$, we move this tile three unit steps to the left.
Afterwards, we move downwards following the tiles and bridges until we reach a free position and place a tile, denoting a copy of $t$.
when we reach the bottom of a column, we remove any tile representing a bridge tile and proceed with the next column (see Fig.~\ref{fig:copy_bridge} right).

\begin{figure}[t]
	\centering
	\includegraphics[angle=90, width=.9\columnwidth]{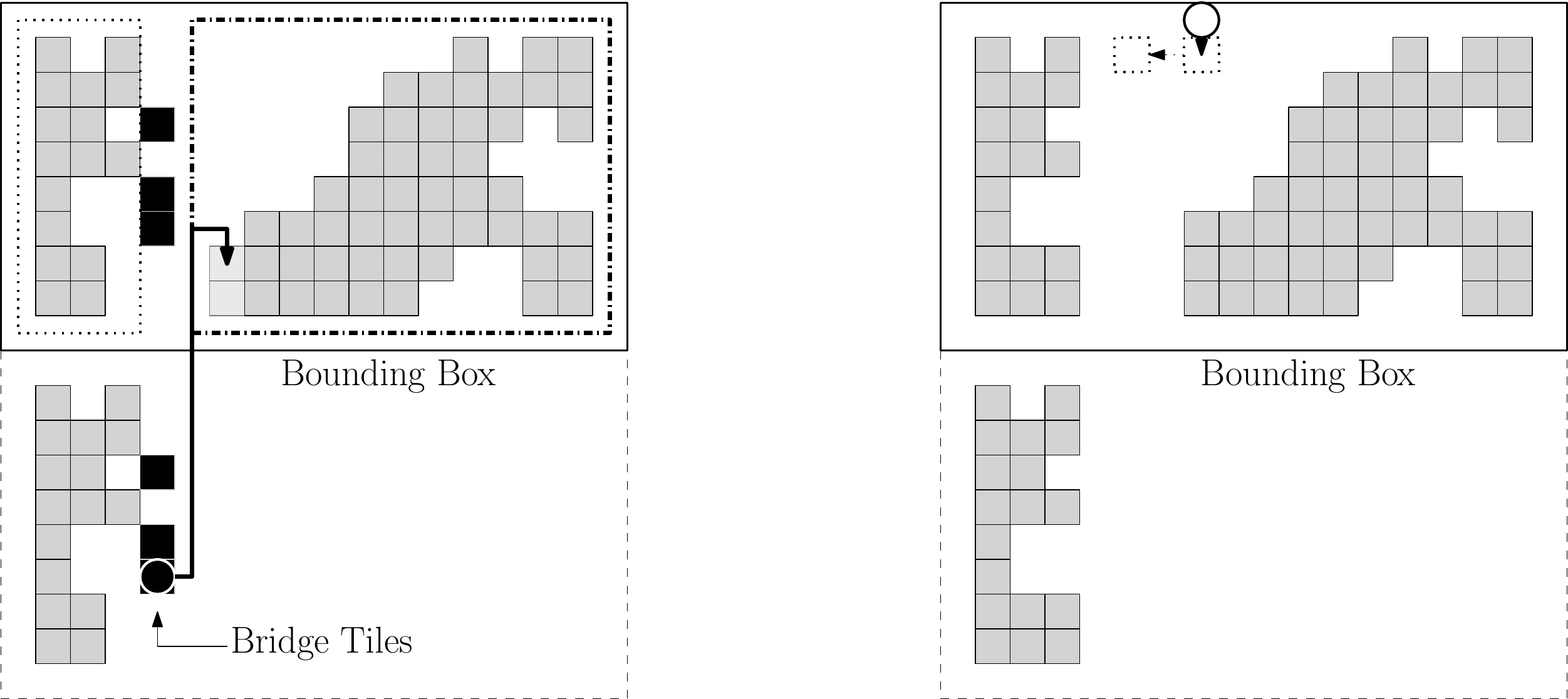}
	\caption{Left: Intermediate step while copying a column of a polyomino (gray tiles) with bridges (black tiles). Tiles in the left box (dotted) are already copied, the tile of $P$ in the right box (dash-dotted) are not copied yet. The robot ($\bigcirc$) moves to next pixel. Right: When the column is copied the bridges get removed and the robot proceeds with the next column.}
	\label{fig:copy_bridge}
\end{figure}

\begin{theorem}\label{th:copy}
	Copying a polyomino $P$ column-wise can be done within $\OCal(wh^2)$ unit steps using $\OCal(N)$ of auxiliary tiles and $\OCal(h)$ additional space.
\end{theorem}

\begin{proof}
	Consider the strategy described above.
	It is straightforward to see that by maintaining the bridges, we can ensure that each tile of $P$ is copied to the correct position.
	As there are $\OCal(wh)$ many pixels that we have to copy with cost of $\OCal(h)$ per pixel, the strategy needs $\OCal(wh^2)$ unit steps to terminate.
	
	Now consider the number of auxiliary tiles.
	We need $N$ tiles for the copy and $\OCal(h)$ tiles for bridges, which are reused for each column.
	Thus, we need $\OCal(N)$ auxiliary tiles in total.
	Because we place the copied version of $P$ beneath $P$, we need $\OCal(h)$ additional space in the vertical direction.
\end{proof}

\subsection{Reflecting a Polyomino}
In this section we show how to perform a vertical reflection on a polyomino $P$ (a horizontal reflection is done analogously).
Assume that we already built the bounding box.
Then we shift the bottom side of the bounding box one unit down, such that we have two units of space between the bounding box and $P$. 
We start with the bottom-most row $r$ and copy $r$ in reversed order beneath the bounding box using bridges, as seen in the previous section.
After the copy process, we delete $r$ from $P$ and shift the bottom side of the bounding box one unit up.
Note that we still have two units of space between the bounding box and $P$.
We can therefore repeat the process until no tile is left.

\begin{theorem}
	Reflecting a polyomino $P$ vertically can be done in $\OCal(w^2h)$ unit steps, using $\OCal(w)$ of additional space and $\OCal(w)$ auxiliary tiles.
\end{theorem}
\begin{proof}
	For each pixel within the bounding box we have to copy this pixel to the desired position.
	This costs $\OCal(w)$ steps, i.e., we have to move to the right side of the boundary, move a constant number of steps down and move $\OCal(w)$ to the desired position.
	These are $\OCal(w)$ steps per pixel,
	$\OCal(w^2h)$ unit steps in total.
	
	It can be seen in the described strategy that we only need a constant amount of additional space in one dimension because we are overwriting the space that was occupied by $P$. 
	This implies a space complexity of $\OCal(w)$.
	Following the same argumentation of Theorem~\ref{th:copy}, we can see that we need $\OCal(w)$ auxiliary tiles.
\end{proof}
\begin{corollary}
	Reflecting a polyomino $P$ vertically and horizontally can be done in $\OCal(wh(w+h))$ unit steps, using $O(w+h)$ additional space and $O(w)$ auxiliary tiles.
\end{corollary}

\subsection{Rotating a Polyomino}
Rotating a polyomino presents some additional difficulties, because the dimension of the resulting polyomino has dimensions $h\times w$ instead of $w\times h$.
Thus, we may need non-constant additional space in one dimension, e.g., if one dimension is large compared to the other dimension.
A simple approach is to copy the rows of $P$ bottom-up to the right of $P$.
This allows us to rotate $P$ with $\OCal(wh)$ additional space.
For now, we assume that $h\geq w$.

We now propose a strategy that is more compact.
The strategy consists of two phases: First build a reflected version of our desired polyomino, then reflect it to obtain the correct polyomino.

After constructing the bounding box, we place a tile $t_1$ in the bottom-left corner of the bounding box, which marks the bottom-left corner of $P$ (see Fig.~\ref{fig:rotate_init}).
We also extend the bounding box at the left side and the bottom side by six units.
This gives us a width of seven units between the polyomino and the bounding box at the left and bottom side.
Now we can move the first column rotated in clockwise direction five units below $P$ (which is two units above the bounding~box), as follows.

We use two more tiles $t_c$ denoting the column, at which we have to enter the bounding box, and $t_2$ marking the pixel that has to be copied next (see Fig.~\ref{fig:rotate_init}).
We can maintain these tiles as in a copy process to always copy the next pixel:
If $t_2$ reached $t_1$, i.e., we would place $t_2$ on $t_1$, then we know that we have finished the column and proceed with the row by moving $t_2$ to the right of $t_1$.
We can copy this pixel to the desired position again by following the bridges and tiles that now make a turn at some point (see Fig.~\ref{fig:rotate_init} for an example).
Note that we cannot place the row directly on top of the column or else we may get conflicts with bridges. We therefore build the row one unit to the left and above the column.
Also note that during construction of the first column or during the shifting we may hit the bounding box. 
In this case we extend the bounding box by one unit.

After constructing a column and a row, we move the constructed row one unit to the left and two units down, and we move the column one unit down and two units left.
This gives us enough space to construct the next column and row in the same way.

When all columns and rows are constructed, we obtain a polyomino that is a reflected version of our desired polyomino. 
It is left to reflect horizontally to obtain a polyomino rotated in counter-clockwise direction, or vertically to obtain a polyomino that is rotated in clockwise direction.
This can be done with the strategy described above.
\begin{figure}
	\centering\hfill
	\includegraphics[width=0.3\columnwidth]{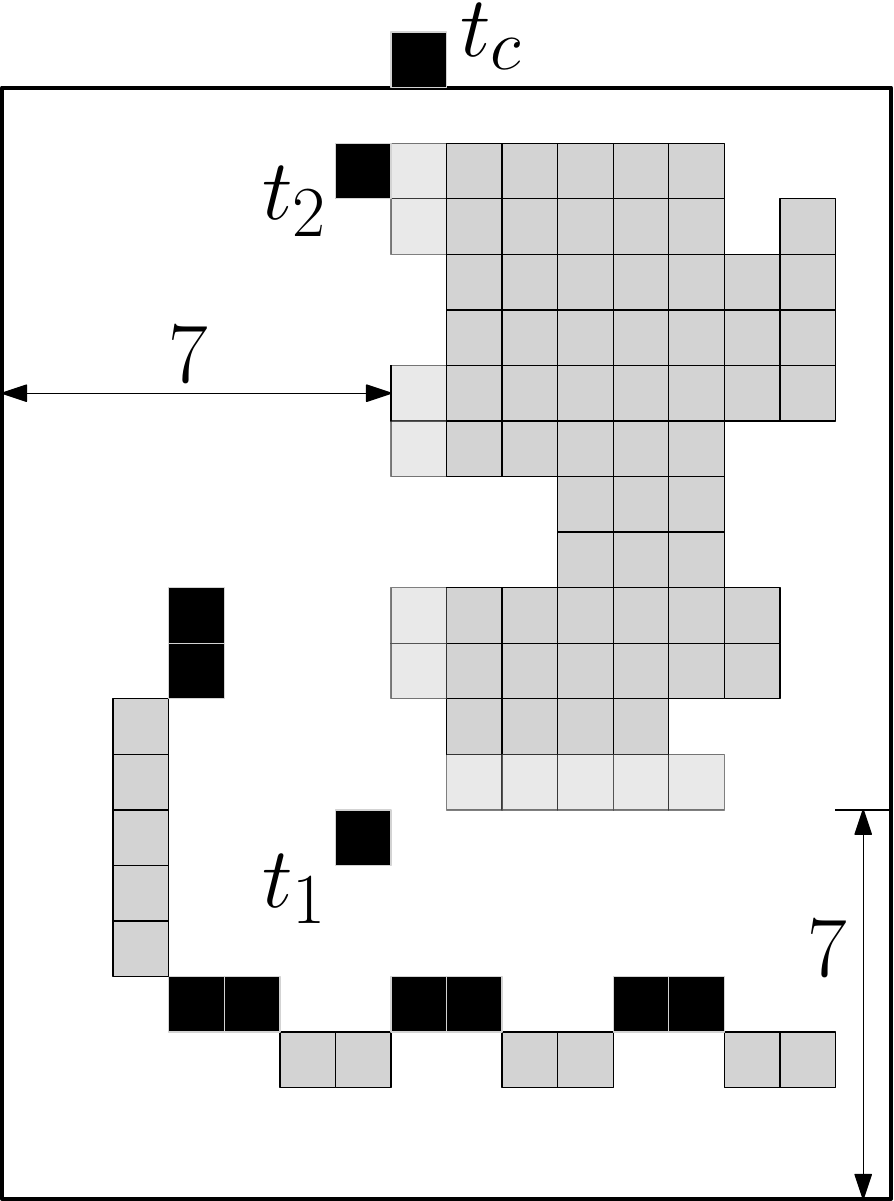}\hfill
	\includegraphics[width=0.3\columnwidth]{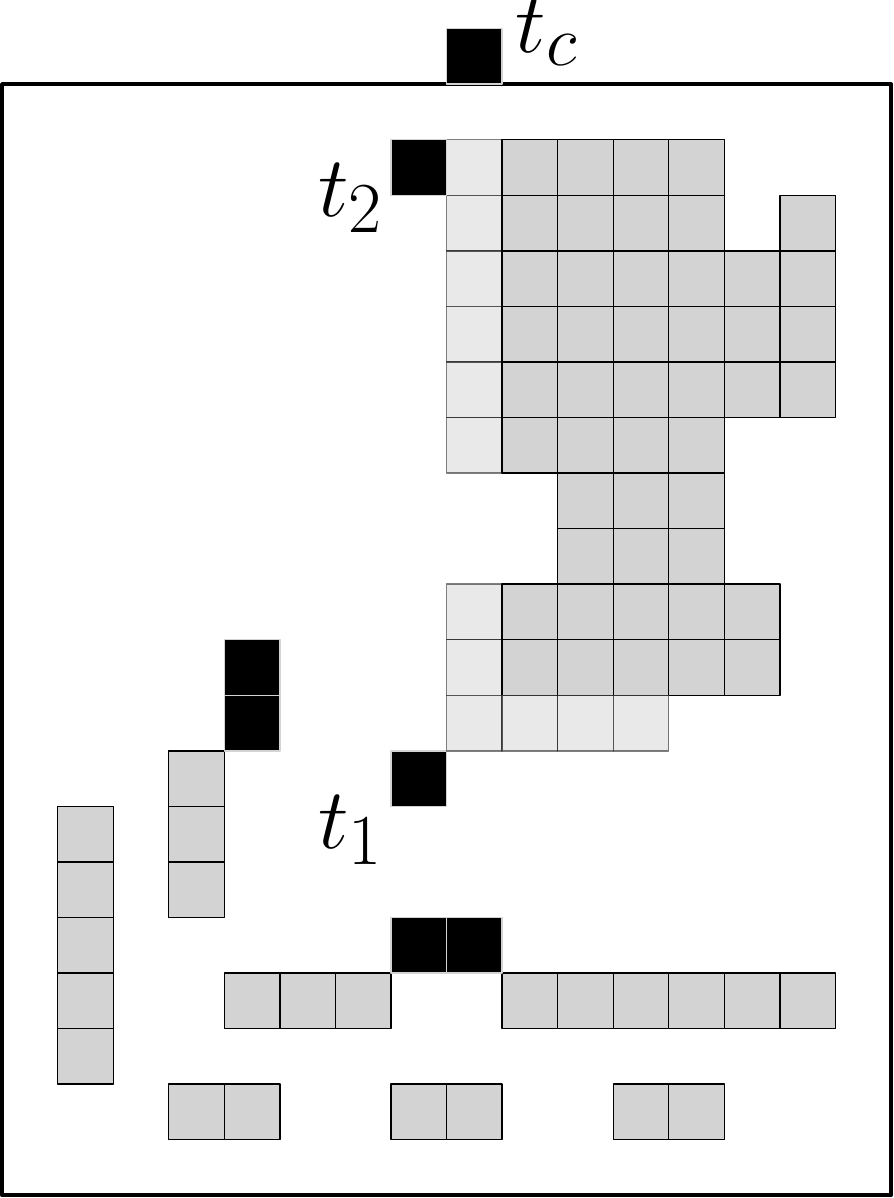}\hfill
	\includegraphics[width=0.3\columnwidth]{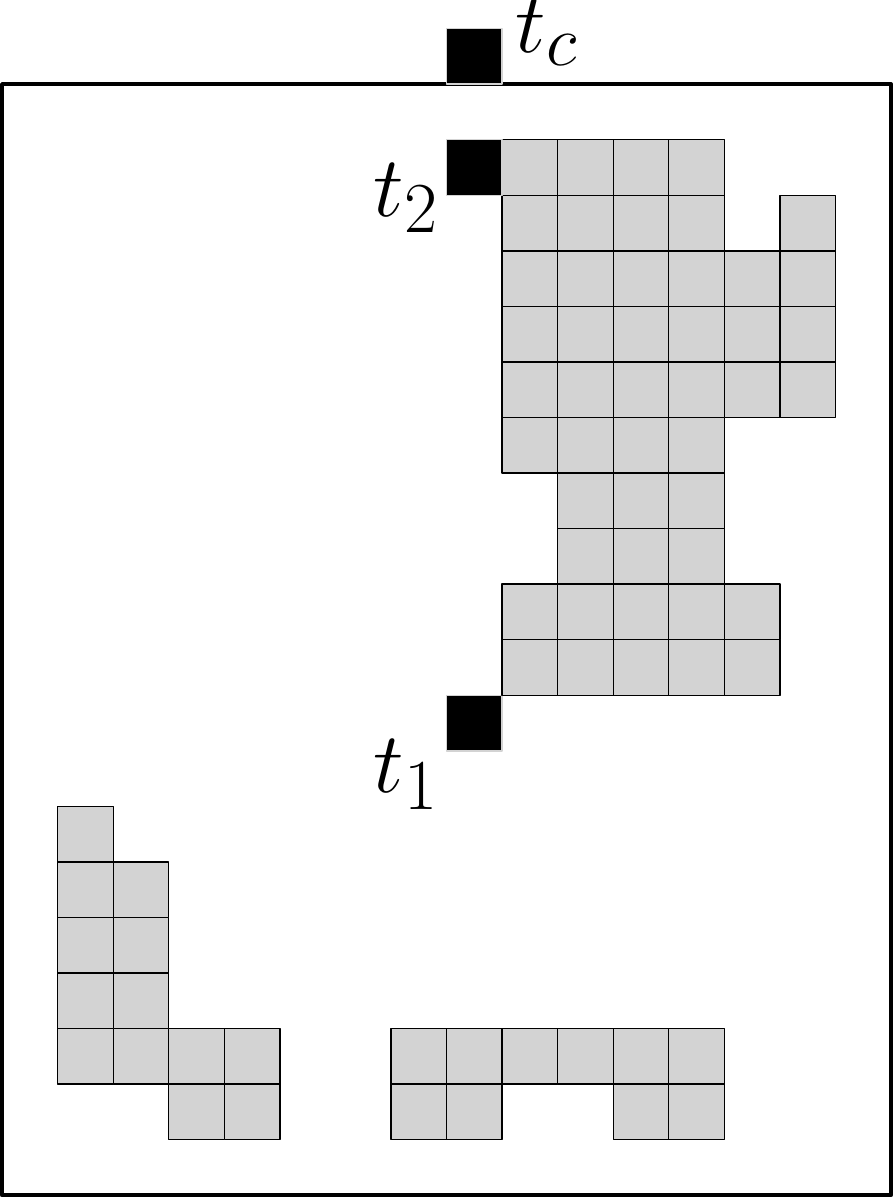}\hfill\phantom{}
	\caption{Left: Constructing the first column and row (light gray in the polyomino). Middle: First row and column have been moved away by a constant number of steps to make space for the next row and column. Right: Merging the first two columns and rows.}
	\label{fig:rotate_init}
	\vspace{-0.4cm}
\end{figure}

\begin{theorem}
	There is a strategy to rotate a polyomino $P$ by $90^\circ$ within $\OCal((w+h)wh)$ unit steps, using $O(|w-h|h)$ of additional space and $O(w+h)$ auxiliary tiles.
\end{theorem}

\begin{proof}
	Like in other algorithms, the number of steps to copy the state of a pixel to the desired position is bounded by $\OCal(w+h)$.
	Shifting the constructed row and column also takes the same number of steps.
	Therefore, constructing the reflected version of our desired polyomino needs $\OCal((w+h)wh)$ unit steps.
	Additionally we may have to extend one side of the bounding box. 
	This can happen at most $\OCal(|w-h|)$ times, with each event needing $\OCal(h)$ unit steps.
	Because $\OCal(|w-h|)$ can be bounded by $\OCal(w+h)$, this does not change the total time complexity.
	
	Because the width of the working space increases by $\OCal(|w-h|)$ and the height only increases by $\OCal(1)$, we need $\OCal(|w-h|h)$ of additional space.
	It is straightforward to see that we need a total of $\OCal(\max(w+h))$ auxiliary tiles.
\end{proof}
\subsection{Scaling a Polyomino}
Scaling a polyomino $P$ by a factor $c$ replaces each tile by a $c\times c$ square.
This can easily be handled by our robot.

\begin{theorem}
	Given a constant $c$, the robot can scale the polyomino by $c$ within $\OCal((w^2+h^2)c^2 N)$ unit steps using $\OCal(c^2wh)$ additional tiles and $\OCal(c^2N)$ additional space.
\end{theorem}

\begin{proof}
	After constructing the bounding box, we place a tile denoting the current column.
	Suppose we are in the $i$-th column $C_i$.
	Then we shift all columns that lie to the right of $C_i$ by $c$ units to the right with cost of $\OCal((w-i)\cdot \sum_{j=i+1}^h cN_{C_j})\subseteq \OCal(wcN)$, where $N_{C_j}$ is the number of tiles in the $j$th column.
	Because $c$ is a constant, we can always find the next column that is to be shifted.
	Afterwards, we copy $C_i$ exactly $c-1$ times to the right of $C_i$, which costs $\OCal(N_{C_i}c)$ unit steps.
	Thus, extending a single row costs $\OCal(c\cdot(wN+N_{C_i}))$ and hence extending all rows costs $\OCal(cw^2N)$ unit steps in total.
	
	Extending each row is done analogously. 
	However, because each tile has already been copied $c$ times, we obtain a running time of $\OCal(c^2h^2 N)$.
	This implies a total running time of $\OCal((w^2+h^2)c^2N)$.
	The proof for the tile and space complexity is straightforward.
\end{proof}

%% file: 07-conclusion.tex
\section{Conclusion}

We have given a number of tools and functions for manipulating an arrangement of tiles by a single robotic automaton.
There are various further questions, including more complex operations, as well as sharing the work between multiple
robots. These are left for future work.